\documentclass[12pt]{article}
\usepackage{amsmath,amsthm,amsfonts,amssymb, graphicx,epsfig}
\usepackage{bbm}

\usepackage[usenames]{color}
\usepackage[active]{srcltx}

\newcommand{\ket}[1]{\left| #1 \right\rangle}
\newcommand{\bra}[1]{\left\langle #1 \right|}

\newcommand{\pd}{\partial}

\newcommand{\Ran}{\operatorname{\mathrm{Ran}}}
\newcommand{\Ker}{\operatorname{\mathrm{Ker}}}
\newcommand{\ran}{\Ran}
\renewcommand{\ker}{\Ker}

\renewcommand{\d}{\mathrm{d}}
\def\tr{\mathrm{Tr}}
\def\slim{\mathop{\mathrm{s-}}\!\lim}
\newcommand{\laplace}{\triangle}
\newcommand{\e}{\varepsilon}
\newcommand{\dbar}{\kern-.1em{\raise.8ex\hbox{ -}}\kern-.6em{d}}

\def\id{\mathbbm{1}}
\def\lin{{\cal L}}
\def\qran{{\cal Q}}
\def\pker{{\cal P}}

\newtheorem{thm}{Theorem}

\newtheorem{prop}[thm]{Proposition}
\newtheorem{exa}{Example}
\newtheorem{rem}{Remark}

\newtheorem{assumption}{Assumption}

\newcommand{\comment}[1]{}

\def \be{\begin{equation}}
\def \ee{\end{equation}}
\def \bea{\begin{eqnarray}}
\def \eea{\end{eqnarray}}

\newcommand{\red}[1]{{\color{red} #1}}
\definecolor{light}{gray}{.75}

\begin{document}

\title{Adiabatic response  for Lindblad dynamics }
\author{J.E. Avron
\\
\small{Department of Physics, Technion, 32000 Haifa, Israel}
\\ M.~Fraas, G.M. Graf
\\
\small{Theoretische Physik, ETH Zurich, 8093 Zurich, Switzerland} }

\maketitle
\begin{abstract}
We study the adiabatic response of open systems governed by Lindblad evolutions. In such systems, there is an ambiguity in the assignment of observables to fluxes
(rates) such as velocities and currents.  For the appropriate notion of flux, the formulas for the transport coefficients are simple and explicit and are governed by the parallel transport on the manifold of instantaneous stationary states. Among our  results we show that the response coefficients of open systems, whose stationary states are projections, is given by the adiabatic curvature. 
\end{abstract}

\section{Introduction}


We are interested in extending the theory of adiabatic response of quantum systems undergoing unitary evolution \cite{BerryRobbins,Thouless83}  to open (quantum) systems governed by Lindblad evolutions. In particular, we are interested in a geometric interpretation of the response coefficients.

In open systems there is usually some choice in setting the boundary between the system and the bath. Setting the boundary fixes the tensor product structure  ${\cal H}_s\otimes {\cal H}_b$. 
Choosing a boundary still leaves a residual ambiguity in observables. For example, given a joint Hamiltonian $H$ of the system and the bath, there is no unique way of assigning to $H$ an observable of the form  $H_s\otimes \id$ describing the energy of the system alone.
\begin{exa}[Lamb shift]\label{lamb}
The interaction of an atom with the photonic vacuum has
two effects on the atom: It leads to decay  and  to ``Lamb shift'' of
the energy levels.  One can choose whether to incorporate the Lamb
shift in the energy of the atom or in its interaction with the bath.
\end{exa} 
Additional  ambiguity arises when considering the \emph{flux}
(rate) of an observable $X$. Any assignment $X\mapsto X_s$ 
of system observables $X_s$
to joint observables $X$ is incompatible with dynamics, if the bath and
the system interact. In fact it is generally impossible to satisfy
both requirements $X_s\otimes \id\mapsto X_s$ and $\dot
X=i[H,X]\mapsto\dot X_s=i[H_s,X_s]$.

The ambiguity in fluxes is physical and plays a key role in this work.  Consider, for example,  damped harmonic motion.  By Newton, the flux of the momentum
is the total force. This force is related to two other forces in this problem:
$$
\mbox{Momentum flux} = \mbox{Spring force} + \mbox{Friction force}.
$$
The momentum flux can be determined from the trajectory of the particle;  The spring force from the force acting on the spring anchor and the friction from the momentum transfer to  the bath.
All these forces have physical significance and are associated with different measurements. In this work we shall focus on observables that are the analog of the momentum flux.  This point of view has been emphasized in the works of \cite{Bel,Car}.

 We study open systems  described by Lindbladians  \cite{Davies,BreuerPet,GoriniKossakowski,Lindblad76}. This framework is usually viewed as giving a simplified, often effective, but  approximate description, of the Hamiltonian dynamics of a system interacting with a bath \cite{Spohn,DaviesSpohn, LebowitzSpohn, JaksicPillet}. 
Model Lindbladians can  be derived from  a Hamiltonian in the ``weak coupling limit'' provided the bath is memory-less (Markovian)  \cite{JoyePillet,Davies}. However, it is also possible to view Lindbladians from a broader perspective, and this is the point of view we  take in this work, namely  as the infinitesimal generators of state preserving maps \cite{Davies}. As such, they provide a natural description of general quantum evolutions, in their own right. 


The Lindblad operator, denoted $\lin$,  is made of a self-adjoint $H$
representing the ``energy'' of the system and a collection of operators,
$\{\Gamma_\alpha\}$, representing the coupling to the bath.  The notion of
``energy'' is, as we have noted in Example \ref{lamb}, ambiguous and this  is manifested in the non-uniqueness of $\{H, \Gamma_\alpha\}$.  
A choice  $\{H,\,\Gamma_\alpha\}$ will be called a \emph{gauge}. Different gauges generate the same dynamics.

We shall consider parametrized Lindbladians, $\lin_\phi$, where the (classical) parameters $\phi\in\cal M$ are viewed as controls. This means that $\{H(\phi),\Gamma_\alpha(\phi)\}$ are functions of the controls\footnote{The functional dependence of a super operator it is indicated by its subscript; that of an operator or a state by its argument.}.  $\cal M$ is the control space, see Fig.~\ref {fig:controlSpace}. 

The main focus of this work is the development of an
adiabatic\footnote{The notion of adiabaticity is contingent on a gap
  condition, Assumption \ref{ass:gap} below.} response theory for the fluxes; Namely, observables of the form $\dot X=\lin^*_\phi X$.  This notion  is gauge invariant (independent of the choice $\{H,\Gamma_\alpha\}$). It turns out that the adiabatic response of such observables has several simplifying features.

A different perspective on the choice of observables comes from a gauge invariant formulation of the \emph{principle of virtual work}.
For isolated systems the principle of virtual work assigns the observable $\partial_\mu H$ with the variation of the $\mu$-th control $\phi^\mu$. Since $\delta H$ is gauge dependent in the Lindblad setting,
formulating a gauge invariant notion of the  principle of virtual work requires the joint variation $\{\delta H,\delta \Gamma_\alpha\}$.   

\begin{figure}[htbp]
\begin{center}
\includegraphics[width=0.5\textwidth]{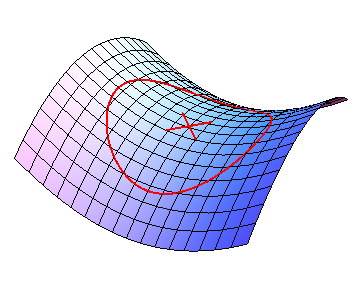}
\caption{The surface represents the control space $\cal M$ and the closed curve represents a closed path  in the space of controls.  The cross denotes a puncture in $\cal M$ where the adiabatic theory fails  and the manifold of stationary states is singular. }
\label{fig:controlSpace}
\end{center}
\end{figure}

Consider a path in control space which is traversed adiabatically,  Fig.~\ref{fig:controlSpace}.
It is a feature of adiabatic evolutions \cite{AFGG} that stationary states evolve by parallel transport within  the manifold of (instantaneous) stationary states. 
The response of fluxes is special in that it is fully determined by the parallel transport of the stationary states, (Theorem \ref{thm:fluxes}).  This is the key to the geometric interpretation of the transport coefficients in linear response.

Parallel transport captures the geometric aspects of adiabatic evolution. However, as a practical method of calculation of transport coefficients, it suffers, in general, from its reliance on solving differential equations. 
A simplification occurs  for  (generic) Lindbladians where the (instantaneous) stationary state is unique
and for dephasing Lindbladians  (where the stationary states coincide with the eigenstates of the Hamiltonian \cite{AFGG}).
In both cases parallel transport is determined algebraically, without recourse to solving differential equations.  As a consequence, the transport coefficients
are geometric and explicit.

In general, the response coefficients of open and closed systems are different. 
One would like to identify those transport coefficients that are immune to certain mechanism of decoherence and dephasing. 
For observables of the form $\lin^*X$  the response coefficients depend on the manifold of stationary states (but not on the underlying dynamics). 
Immunity then follows whenever
the stationary states are unaffected by decoherence and dephasing.
This is the case for two physically interesting families of Lindbladians:  Dephasing Lindbladians and Lindbladian which allow for decay to the (Hamiltonian) ground state.


\section{Lindbladians}
\label{sec:lindbladian}

The Lindblad (super)\footnote{Super operators will be denoted by script characters.} operator \cite{Davies,BreuerPet} is given formally by\footnote{The normalization differs by factor 2 from that of  \cite{Lindblad76}.} 
	\be
	 \lin\rho= -i [H,\rho] + {\cal D}\rho, \quad  {\cal D}\rho=\sum_\alpha 2 \Gamma_\alpha \rho \Gamma^*_\alpha
	- \Gamma^*_\alpha \Gamma_\alpha \rho - \rho \Gamma^*_\alpha \Gamma_\alpha,\quad 
	\end{equation}
where the state $\rho$ is trace class. The Hamiltonian part $H$ is self-adjoint (and local). The $\Gamma_\alpha$ are essentially arbitrary, but finitely many for simplicity. 
Models describing exchange of energy involve non-self-adjoint $\Gamma_\alpha$ while models of measurement involve $\Gamma_\alpha$ which are spectral projections  (non-local in general).
  
$\lin$ is the generator of state and trace preserving contractions. The dual (super) operator $\lin^*$ acts on the space of bounded operators, this being the dual of the space of trace class operators. 
When ${\cal D}=0$ the evolution is unitary.
 To avoid technical difficulties with unbounded operators we shall  assume:

\begin{assumption}\label{bounded}
$H$ and $\Gamma_\alpha$ are bounded operators.
\end{assumption}
The assumption implies that a duality relation,
\begin{equation}
\label{eq:duality}
\tr(X \lin \rho) = \tr((\lin^*X) \rho)\,,
\end{equation}
holds for all states $\rho$ and bounded operators $X$.
We shall occasionally consider standard physical examples with unbounded operators. However all these 
examples are simple enough that one can check that formal manipulations are indeed justified. 
We shall study the time evolution of the state $\rho$: 
	\be\label{lindbladian0}
	\dot\rho=\lin\rho, \quad \Bigl(\dot\rho=\frac{d\rho}{dt}\Bigr). 
	\ee	
 It is convenient to introduce a notation that distinguishes stationary states from general states. We shall denote stationary states by $\sigma$, namely  $\lin\sigma=0$  (with $\tr \,\sigma=1$).  The (super) projection   on the stationary states shall be denoted by $\pker$, so $\pker \sigma=\sigma$.

\subsection{Gauge transformations}\label{sec:ambiguity}
 $\lin$  does  not determine $\{H,\Gamma_\alpha\}$. In fact, $\lin$ is invariant under the joint variation \cite{BreuerPet}
\begin{gather}
\label{gauge}
 H\mapsto H+e\id-i\sum_{\alpha}(g_\alpha^*\Gamma_\alpha-g_\alpha\Gamma_\alpha^*), ~ \Gamma_\alpha\mapsto\Gamma_\alpha+ g_\alpha \id, ~ (g_\alpha \in \mathbb{C},\, e \in \mathbb{R}).
\end{gather}
Moreover, $\Gamma$ and ${\cal U}\Gamma$ represent the same Lindbladian when ${\cal U}$ is unitary in the sense that
	\be\label{gauge2}
	    ({\cal U}\Gamma)_\alpha=
	    \sum_\beta {\cal U}_{\alpha \beta} \Gamma_\beta,\quad {\cal U}^{-1}={\cal U}^*.
	 \ee
We shall refer to the freedom in  $\{H,\Gamma_\alpha\}$ as {\em gauge freedom}.
The observable $H$, which one would like to interpret as the energy of the system, is therefore ambiguous a priori\footnote{Interferometry allows  to compare the evolution in one arm of the interferometer with a different evolution in the other.  This can be used to fix some of  the gauge freedom. Interferometry for open system is described in \cite{Sjoqvist}. }. This ambiguity does not go away by considering explicit physical models weakly  coupled to a bath, as Example~\ref{lamb} shows. In the examples that we consider, we pick a natural gauge.  

 \subsection{Lindbladians with a unique stationary state}
 \label{sec:unique}
 Generic finite dimensional Lindbladians have a unique stationary state $\sigma$.
The (super) projections $\pker$ on the  stationary state is given by
	\be\label{UGS}
	\pker\rho = \sigma\, \tr\, \rho, \quad (\tr\,\sigma=1).
	\ee
Evidently $\pker$ is trace preserving and $\pker^2=\pker$ (since $\tr
\,\sigma=1$). It is {\em not orthogonal}, not even formally. In fact,
the dual projection $\pker^*$, which acts naturally on observables,
is given by a different expression, 
	\be\label{NonOrth}
	\pker^*X= \id \cdot \tr(X\sigma).
	\ee
A basic identity we shall need is
	\be\label{LP}
	\lin\pker =0=\pker\lin.
	\ee
The first equality is evident. The second follows from the trace preserving property of $\lin$
	\[
	 \pker \lin \rho= \sigma\tr(\lin\rho)=0.
	\]	
We shall denote by  $\qran$ the complementary projection	
	\be
	\pker+\qran=\id.
	\ee
Evidently,
	\be\label{LQ}
	\lin\qran =\lin=\qran\lin.
	\ee
This will play a role in the sequel.

\subsection{Dephasing Lindbladians}\label{sec:dephasing}
Dephasing Lindbladians are intermediate between Hamiltonians and Lindbladians with a unique stationary state. They are  characterized \footnote{When $H$ is simple, the characterization can be phrased as a commutation condition $[H,\Gamma_\alpha]=0$. The two characterizations differ e.g. in the case  $H=\id$ and the choice we make guarantees that dephasing Lindbladian share the stationary states with the Hamiltonian.   } by $\Gamma_\alpha = \Gamma_\alpha(H)$ for some functions $\Gamma_\alpha$. In particular,
\be\label{dephasing}
\lin P=0=i[H,P],
\ee
where  $P$ is a  spectral projection for $H$.

When $H$ is finite dimensional,
all spectral projections $P_j$ are finite dimensional and dephasing Lindbladians share the stationary states with the Hamiltonian. 
 The manifold of stationary states is then the span of the $P_j$. The (super) projections $\pker$ on this manifold and its complement $\qran$, are given by 
	\be\label{PD}
	\pker\rho= \sum_j P_j \rho P_j, \quad \qran\rho=\sum_{j\neq k}  P_j \rho P_k.
	\ee  
$\pker$ and $\qran$ are orthogonal projections, in the sense that their adjoints are given by the same expressions. 
$\pker$ satisfies Eq.~(\ref{LP}) and $\qran$ satisfies Eq.~(\ref{LQ}).


\section{Stationary states}

 \begin{figure}[htbp]
\begin{center}
\includegraphics[width=0.5\textwidth]{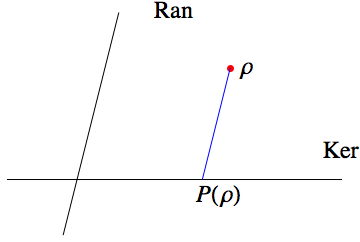}
\caption{The  (super) projection $\pker$ projects on $\Ker \lin$ along $\Ran \lin$.  The two spaces are transversal in the Lindblad setting. }
\label{fig:projection}
\end{center}
\end{figure}
Let us consider the (super) projections $\pker$ on the manifold of stationary states from a perspective that puts the special classes treated above in a uniform context.

It is a basic property of Lindblad operators that kernel and range are
transversal \cite{AFGG}, 
%
\begin{equation}
\label{geometricgap}
\ker \lin \cap \ran \lin  = \{0\}.
\end{equation}
This follows from $\exp(t\lin)$ being a contraction: $\lin^2 \rho = 0$ implies $\exp(t\lin)\rho = \rho + t\lin\rho$
and we conclude $\lin\rho = 0$.

This allows to define a projection on the direct sum $\ker\lin \oplus \ran\lin$ by
\begin{equation}
\label{GeoProjection}
\pker\rho = \left\{ \begin{array}{lcr}
			\rho & \mbox{when} & \rho \in \ker\lin \\
			0    & \mbox{when} & \rho \in \ran\lin.
		\end{array}
	\right.
\end{equation}

\begin{assumption}[Gap condition]\label{ass:gap}
\hfill
$0$ is an isolated point in the spectrum of $\lin$ and $\pker$ is given by the Riesz projection
	\be \label{PKer}
	\pker= \frac 1 {2\pi i} \oint \frac {dz} {z-\lin}, 
	\ee
where the contour encircles $0$ but no further points of the spectrum
(see Fig.~\ref{fig:Riesz}). 

\end{assumption}

\begin{rem}
If $0$ is an eigenvalue of finite algebraic multiplicity, then the Riesz projection part is for free. The assumption is satisfied when $H,\,\Gamma_\alpha$ are finite dimensional.
The assumption guarantees that 
$\ran \lin$ is a closed subspace. 
\end{rem}
\begin{figure}[htbp]
\begin{center}
\includegraphics[width=0.3\textwidth]{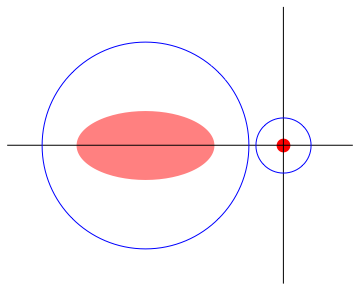}
\caption{Since $\lin$ is a contraction, its spectrum away from the origin is contained in the ellipsoid blob in the left half-plane. The origin is assumed to be  an isolated point in the spectrum. The circles are the integration contour for the Riesz projection $\pker$ and $\qran$.}
\label{fig:Riesz}
\end{center}
\end{figure}

The consistency of Eqs.~(\ref{GeoProjection}) and (\ref{PKer}) deserves a discussion.
In fact, the Riesz projection $\pker$ always satisfies the first line of Eq.~(\ref{GeoProjection}) and the validity of the second one is the core of the assumption. 
To see this consider, besides of $\pker$ given by Eq.~(\ref{PKer}), also $\qran$ similarly given in terms of a contour encircling the complementary part of the spectrum. Then
	\be\label{comm}
	\pker+\qran =\id, \quad [\pker,\lin]=[\qran,\lin]=0,
	\ee
proving the first line.
Now assuming the second line of Eq.~(\ref{GeoProjection}),
the eigenspace associated to $0$ has a trivial Jordan block. 
 This means that the Laurent expansion of the resolvent does not have a $z^{-2}$ term and hence
	\be\label{PL}
	\pker\lin =\frac 1 {2\pi i} \oint dz \frac {z}{z-\lin} =0
	\ee
and\footnote{We recall that ${(\cdot )}^*$ is the Banach space notion of dual.}
	\be\label{QL}
	 \qran\lin=\lin, \quad \lin^*=\lin^*\qran^*.
	\ee
This places Eqs.~(\ref{LP}, \ref{LQ}) into their general natural context.
 
\begin{exa}[Gapless Lindbladians]
\label{exa:ConSpec}
It may happen that $H$ is gapped and $\lin$ is gapless. For example, let $H$ be a Hamiltonian with a ground state separated by a gap from a continuous spectrum. Then the associated 
Lindbladian $\lin \rho = -i[H,\rho]$ has the eigenvalue 0  embedded in the continuous spectrum.

\end{exa}

\section{Controlled Lindbladians}\label{sec:control}

The geometric aspects emerge when one turns one's attention to a parame\-tri\-zed family of Lindbladians  $\lin_\phi$. We shall call the parameters $\phi\in\cal M$ {\em control} and   $\cal M$ the control space. This makes the Hamiltonian $H(\phi)$ and the coupling to the bath $\Gamma_\alpha(\phi)$ functions of the controls. The explicit form of these functions is, of course,  model specific.

\begin{assumption}[Controlled Lindbladians]\label{ass:smoothness}\hfill

\begin{enumerate}
\item[(A)] The Lindbladian $\lin_\phi$ is a bounded (super) operator which is a smooth function of the controls $\phi$.
\item[(B)] The gap condition, Assumption~\ref{ass:gap}, holds for all $\lin_\phi$. 
\end{enumerate}
\end{assumption}

We shall call the stationary states of $\lin_\phi$ the {\em instantaneous stationary states}. 

\subsection{Iso-spectral Lindbladians}\label{sec:unitaries}

A distinguished family of controlled Lindbladians is the family of iso-spectral Lindbladians  given by the action of unitaries on $H$ and $\Gamma$:
	\be\label{isoLin}
	H(\phi)=U(\phi) H U^*(\phi), \quad \Gamma_\alpha(\phi)=U(\phi) \Gamma_\alpha U^*(\phi).
	\ee
The Lindbladian describing  a harmonic oscillator coupled to a thermal bath, whose anchoring point is  controlled, is an example: 
\begin{exa}[Controlled oscillator in thermal contact] \label{exa:decay}
	A Harmonic oscillator, anchored at the origin and coupled to a heat bath, is described by the Lindbladian 
	\be
	H=a^*a,\quad \Gamma_-= \sqrt\gamma_- a, \quad\Gamma_+= \sqrt\gamma_+ a^*, \quad  (\gamma_->\gamma_+>0) 
	\ee	
and where $\sqrt 2 a={x+ip}$. The stationary state of the oscillator is a thermal state with $\beta=\log (\gamma_-/\gamma_+)$ \cite{BreuerPet}. The Harmonic oscillator with controlled anchoring point is described by  the iso-spectral family with $U(\phi)= e^{-i p\phi}$. Explicitly
	\be\label{OscShift}
	\sqrt 2\, a(\phi)=\sqrt 2\, U(\phi)\, a\, U^*(\phi)=(x-\phi) +ip.
	\ee
Since the $\Gamma$'s adjust to $H$ the oscillator wants to relax to the thermal state of the instantaneous Hamiltonian.
 \end{exa}

\subsection{Parallel transport}

 We shall denote instantaneous stationary states by $\sigma$.  By definition $\pker\sigma=\sigma$ (we allow $\dim\pker\ge 1$). By Assumption~\ref{ass:smoothness}  the projection on the stationary states, $\pker_\phi$, is a smooth projection on control space and $d\pker_\phi$ is a bounded operator valued form. For notational convenience, we henceforth suppress the explicit $\phi$ dependence and write $\pker$ for $\pker_\phi$ etc.

The differential of  $\sigma=\pker\sigma$ gives the identity  $d\sigma= (d\pker)\sigma+ \pker d\sigma$. Since $\pker$ is a projection $\pker (d\pker) \pker=0$ and consequently the projection $\qran d\sigma=(d\pker)\sigma$ is determined while $\pker d\sigma$ is not.
 Parallel transport is the requirement, given in two equivalent forms, 
\be\label{parTran}
\pker d\sigma =0 , \quad d\sigma=(d \pker)\sigma.
\ee
This evolution of $\sigma$ is naturally
 interpreted geometrically as parallel transport\footnote{For a different perspective which focuses on the analogs of Berry's phase see e.g. \cite{GefenMakhlin}.   }:  There is no motion in $\pker$.
The case $\dim\pker=1$
 is a special simple case, in that there is a unique state $\sigma=\sigma(\phi)$ in the range of $\pker(\phi)$. It solves Eq.~(\ref{parTran}) without further ado.

\begin{prop}\label{prop:trace}
The form $d \sigma$ is trace class.
\end{prop} 
\begin{proof} Follows from Eq.~(\ref{parTran}), the fact that $\sigma$ is trace class and that $d\pker$ is bounded by Assumption~\ref{ass:smoothness}. 
\end{proof}


\subsection{Holonomy of parallel transport}\label{sec:holonomy}

In general, parallel transport, Eq.~(\ref{parTran}), does {\em not} integrate to a function on control space $\cal M$ unless the curvature vanishes: $\pker d\pker\wedge d\pker\pker=0$ (see Appendix~\ref{appendix:a}). If such a function $\sigma=\sigma(\phi)$ exists, it will be called an {\em integral of parallel transport}. This is, of course, automatic if either $ {\cal M}=\mathbb{R}$, or $\dim\pker=1$, (see Eq.~(\ref{UGS})). 

Parallel  transport is consistent with the convex structure of stationary states \cite{AFGG}.  As a consequence it preserves  extremal stationary states. 
Recall that a (stationary) state is called {\em extremal}~  if it can not be written as a convex combination of two other (stationary) states.  For such extremal states we have:

\begin{prop}[Parallel transport of extremal states]
\label{extreme}
The parallel transport equation takes extremal stationary state to an extremal stationary state. If, moreover the manifold of (instantaneous) stationary states is a simplex, spanned by a finite number of isolated extremal states $\sigma_j(\phi)$, then the function on $\cal M$ 
	\be\label{eq:sigmaF}
	\sum p_j \sigma_j(\phi)
	\ee
with $p_j$ independent of $\phi$, is an integral of parallel transport. 
\end{prop}

\begin{proof} A more general statement has been proved in \cite[Proposition 3]{AFGG}. The intuition is that states move inside $\ker\lin$ as little as 
possible. In particular the boundary should be mapped by parallel transport to the boundary and extremal points to extremal points.
\end{proof}

Parallel transport is path independent for two important families of Lindbladians:
\begin{itemize}
\item Lindbladians with a unique stationary state, (Section~\ref{sec:unique}). 
\item Dephasing Lindbladians  where the isolated extremal states are  the  one dimensional spectral projections \mbox{$\sigma_j(\phi)=P_j(\phi)$}, (Section~\ref{sec:dephasing}). 
\end{itemize}

Finally we discuss the parallel transport for iso-spectral families.
\begin{prop}
\label{IsoParTran}
The family $\sigma(\phi) = U(\phi)\sigma U^*(\phi)$ is an integral of parallel transport, Eq.~(\ref{parTran}), if and only if 
\begin{equation}\label{consistencyG}
\pker[G,\,\sigma] = 0,
\end{equation}
where $G d \phi = iU^*d U$. In particular the condition applies when $\sigma$ is an isolated extremal point.
\end{prop}

\begin{proof} Condition (\ref{consistencyG}) follows by inserting 
\begin{equation}
\label{ULocalResponse}
d \sigma =  -i[ G, \sigma] d \phi
\end{equation}
 into Eq.~(\ref{parTran}). The last claim is a consequence of Prop.~\ref{extreme}. 
\end{proof}

\begin{exa}
Condition~(\ref{consistencyG}) holds for any iso-spectral family with a unique stationary state (and $G$ bounded). In this case $\pker$ is given by Eq.~(\ref{UGS})
 and
$$
\pker [G,\,\sigma] = \sigma \tr \left( [G,\,\sigma] \right) = 0
$$
by the cyclicity of the trace.
\end{exa}


\section{Observables and {fluxes}}
\label{sec:ObsRates}

We denote observables by $X$.  The evolution of observables (in the Heisenberg representation) is generated by $\lin^*$: 
	\begin{equation}\label{Heisenberg}
	\dot X= \partial_tX+\lin^* X, \quad \Bigl(\dot X=\frac {dX}{dt}\Bigr)
	\end{equation}
where
\begin{equation}\label{LStar}
	\lin^*X = i [H,X] + {\cal D}^*X, \quad {\cal D}^*X=\sum_\alpha  \Gamma_\alpha^*[ X ,\Gamma_\alpha]+ [ \Gamma_\alpha^*, X ]\Gamma_\alpha.
	\end{equation}
 $\dot X $ is itself an observable: We refer to $\dot X$ either as the {\em flux} (or rate) of $X$ or simply as the flux $\dot X$. 	For example, the velocity is the flux of the position and the force is the flux of the momentum. 
 \begin{assumption}\label{ass:partial}
 $X$ is not explicitly time dependent, i.e. $\partial_tX=0$, and hence $\dot X= \lin^*X$.
 \end{assumption}
Fluxes lie in $\Ran\lin^*$.  
They have the special feature of vanishing expectation in stationary states. In fact:

\begin{prop}[No currents in stationary states]\label{NoCurrents}
Let  $\sigma$ be a trace class stationary state and $\dot X$ the flux of  the  bounded observable $X$. Then, 
$\tr (\dot X\sigma)=0$. Conversely, if $\tr (A \sigma)=0$ for any stationary state, then $A=\lin^*X=\dot X$ for some bounded observable $X$.
\end{prop}
\begin{proof}
The (super) operator $\lin^*$ acts on the space of bounded operators, this being the dual of the space of trace class operators. By Eq.~(\ref{Heisenberg}) and Eq.~(\ref{eq:duality}):
	\begin{align}\label{X0atEq}
	\tr (\dot X\sigma)= \tr((\lin^*X)\sigma)=\tr(X(\lin\sigma))=0\,.
\comment{	=\tr\Bigl((\lin^*Q^*X)\sigma\Bigr) \\
	&=\tr\bigl(\qran^*X(\lin\sigma)\bigr)=0\,.} 
	\end{align}
Conversely when (a bounded)
 $A$ has vanishing expectation in stationary states, then
	\be
	0=\tr (A(\pker\rho))=\tr \bigl((\pker^*A)\rho\bigr)
 	\ee
holds for any $\rho$. Hence $\pker^*A=0$ and $A$ lies in the range of $\lin^*$. 
\end{proof}

An example of an observable which is not a flux is:

\begin{exa}[Loop currents]
Consider a quantum particle on  a ring  with $p$ sites, $\mathbb{Z}_p$, evolving by  the (bounded) Hamiltonian
	\[
	H(\phi )=e^{i\phi}T+e^{-i\phi} T^*,\quad (T\psi)(n)=\psi(n-1).
	\]
$p\phi$ may be interpreted as the magnetic flux threading the ring. The stationary states are $\bra{n}\sigma_j\ket{m}= p^{-1} e^{i(n-m) 2\pi j/p}$, ($j=0,\ldots p-1$). The angular velocity is the (bounded) operator
	\[
	-\partial_\phi H=-i(e^{i\phi}T-e^{-i\phi} T^*),
	\]	
Since $-\tr (\partial_\phi H \sigma_j)= 2\sin(2\pi j/p-\phi)$ does not vanish in stationary states, $\partial_\phi H$ is not a flux: The angle is not an observable, since it is multivalued.
\end{exa}	

Current carrying stationary states can  occur only  when either the system is multiply connected or in the thermodynamic limit. Examples are  supercurrents, where magnetic vortices effectively make the system multiply connected,  and  stationary currents in mesoscopic rings \cite{bloch, bohm}.
%
%
The velocity is the flux of the position operator, which is usually unbounded. It is therefore interesting to examine conditions that would allow extending Prop.~\ref{NoCurrents} to unbounded operators.  Indeed, the gap condition implies that the expectation values of fluxes vanish in stationary states even for unbounded $X$ provided $Q^*X$ is bounded. Indeed, the gap condition allows us to use Eq.~(\ref{QL}) and replace Eq.~(\ref{X0atEq}) by
\be\label{X0atEqUb}
\tr((\lin^*X)\sigma)=\tr((\lin^*Q^*X)\sigma))=\tr\bigl((\qran^*X)(\lin\sigma)\bigr)=0\,.
\ee
A more careful discussion of this point is given in Prop.~\ref{NoCurrentsUnb} of Appendix \ref{unbounded}. 

An example where $X$ is unbounded but $Q^*X$ is bounded is:

\begin{exa}[Taming $X$]\label{taming}
Consider a quantum particle with spin hopping on the integer lattice. The Hilbert space is $\ell(\mathbb{Z})\otimes \mathbb{C}^2 $ and let the Hamiltonian be
\[
	H=T\otimes a+T^*\otimes a^*,
	\]
where $T$ is the unit left shift and $a$ and $a^*$ are the spin lowering and raising operators $a^2=(a^*)^2=0, \ \{a,a^*\}=1 $. Since $H^2=\id$ we can write $H$ as a difference of two  (infinite dimensional) projections: 
	\[
	H =P_+-P_-
	,\quad 2 P_{\pm}= \id\pm H.
	\]
The position operator is $X\otimes \id$. It is clearly unbounded. Since $[T,X]=T$,  the velocity is the bounded operator
	\[
	\dot X= i[H,X\otimes \id ]=i(T\otimes a- T^*\otimes a^*).
	\] 
In fact, $\dot X^2=\id$.  
The appropriate version of Eq.~(\ref{PD}) says that $Q^*$ is given by
	\begin{align*}
	\qran^* X&=P_+X P_-+P_-X P_+\\
	&=P_+[X, P_-]+P_-[X, P_+]\\
	&=\frac 1 2 \left(-P_+[X, H]+P_-[X, H]\right)\\
	&=-\frac i 2\left( P_+-P_-\right)\dot X,
	\end{align*} 
being a product of bounded operators,  it is bounded, even though $X$ is not.
\end{exa}


	
\subsection{Virtual work for Lindbladians}\label{sec:virtual}




The principle of virtual work associates observables with variations of a controlled Hamiltonian $H(\phi)$. Our aim here is to formulate a corresponding principle for Lindbladians. 

 Observe that, first $\lin$ is a (super) operator, so its variation {\em does not} define an observable and second, the notion of ``energy'' is ambiguous in Lindblad evolutions.
The principle of virtual work we formulate is gauge invariant in the sense of  Section~\ref{sec:ambiguity}. 
\begin{thm}\label{virtualVariation}
The observables $X_\mu$ given by 

\begin{equation}
\label{martin}
X_\mu \delta\phi^\mu = \delta H +
i \sum_\alpha \left(\Gamma^*_\alpha \delta \Gamma_\alpha - \delta \Gamma^*_\alpha \Gamma_\alpha \right),
\end{equation}
 involving the joint variation of $H$ and $\Gamma$,
 are (formally) self-adjoint and free from the ambiguity in $H$ and $\Gamma$, 
under $\phi$ independent gauge transformations. 
\end{thm}
\begin{proof}
For $g_\alpha$ and $e$ independent of $\phi$,  the gauge transformation, Eq.~(\ref{gauge}), affects the variation by 
$$
\left(\delta H,\,\delta \Gamma_\alpha \right) \mapsto \Bigl(\delta H - i \sum_{\alpha} (g^*_\alpha \delta \Gamma_\alpha - g_\alpha \delta\Gamma^*_\alpha),\,\delta \Gamma_\alpha \Bigr).
$$
This leaves $X_\mu$ invariant. The same applies to the transformation ${\cal U}$ of Eq.~(\ref{gauge2}).
\end{proof}
\begin{rem}[A second gauge invariant family]
A second family of observables that are gauge invariant is 
$ \delta\left(\sum_\alpha \left[ \Gamma_\alpha,\Gamma^*_\alpha \right]\right)$. 
This follows from the commutativity  $[\Gamma_\alpha,g_\beta \id]=0$, and the unitarity of ${\cal U}$.
\end{rem}

The observables $X_\mu$ extend the notion of the principle of 
virtual work to the Lindbladian setting. The physical interpretation of $ X_\mu$ is often suggested by dimensional analysis. It depends on the choice of controls and  is model dependent. 

\begin{exa}[Controlled oscillator in thermal contact: Example~\ref{exa:decay} continued]\label{exa:oscillator} Shifting the anchoring point of the oscillator gives
	\begin{align*}
	\delta H= -\frac {a+a^*}{\sqrt 2}{\delta\phi}= -x\delta\phi,\quad   i  \bigl(\Gamma^*_\pm \delta \Gamma_\pm - \delta \Gamma^*_\pm \Gamma_\pm \bigr)=\mp\gamma_\pm\frac{a^*-a}{i\sqrt 2}{\delta\phi}=\pm\gamma_\pm p\delta\phi.
	\end{align*}
$-x$ is  the spring force, while $-\gamma_-p$ gives the friction force due to the cold contact and $\gamma_+p$ the gain from the hot contact. 
The observable distinguished by virtual work is the total force i.e. the momentum flux 
	\be\label{totForce}
	\dot p=-x-(\gamma_--\gamma_+)p,\quad (\gamma_->\gamma_+).
	\ee
\end{exa}\label{force}
In the example, the principle of virtual work gives a flux. This is not a coincidence.
For iso-spectral Lindbladians, Eq.~(\ref{isoLin}), virtual work is a flux.  More precisely, 
let $G_\mu$ denote the (local) infinitesimal generators
	\be
	\label{InfGenerator}
	G \delta \phi = G_\mu \delta\phi^\mu =i {U^* \delta U}
	\ee
(summation implied). The variations are:
	\be\label{QG}
	\delta H= i[H,G_\mu]\, \delta\phi^\mu, \quad \delta \Gamma_\alpha= i[\Gamma_\alpha,G_\mu]\, \delta\phi^\mu.
	\ee
\begin{thm}[Virtual work and fluxes] \label{thm:VW}
For iso-spectral families of Lindbladians generated by $G_\mu$, the observables associated with the principle of virtual work,
Eq.~(\ref{martin}), are the fluxes of the generators $G_\mu$ :
	\be\label{Q}
	\lin^*G_\mu\delta\phi^\mu = \delta H +
i \sum_\alpha \left(\Gamma^*_\alpha \delta \Gamma_\alpha - \delta \Gamma^*_\alpha \Gamma_\alpha \right).
	\ee
In particular we have Noether's theorem in the form: If $\delta U$ is a symmetry, in the sense that the r.h.s. vanishes, then its generator is a conserved quantity.
\end{thm}
\begin{proof} By Eqs.~(\ref{martin}) and (\ref{QG}). \end{proof}

\subsection{Currents}\label{sec:currents}
Just as there are three notions of force in a damped oscillator, there  are several notions of currents in an open system.
By partitioning the  the system into a subsystem $\Omega$ and its complement $\Omega^c$, one identifies three notions of currents:
\begin{itemize}
\item   The rate of charge  in the subsystem $\Omega$, 
	\be\label{Irate}
	\dot Q_\Omega=\lin^*Q_\Omega
	\ee

\item For  charge conserving $H$, i.e. $[H, Q_{\Omega\cup\Omega^c}]=0$,  the current $I_{\partial\Omega}$ flowing  from $\Omega$ to its complement is
\be\label{eq:I}
	I_{\partial\Omega}= i[H,Q_\Omega].
	\ee
With $H$ local, this current is naturally associated with the boundary.
\item The current, $S_\Omega$, flowing from the subsystem $\Omega$ to the bath defined via charge conservation
	\be\label{TotQ}
	\dot Q_\Omega=I_{\partial \Omega}+ S_\Omega.
	\ee
\end{itemize}
The three currents are measured by different instruments: $\dot Q_\Omega$ is measured by an electrometer while $I_{\partial\Omega}$ is measured by an ammeter that monitors the flow at the boundary between the subsystems. In view of 
	\be\label{Irate}
	\lin^*Q_\Omega= i [H,Q_\Omega] + {\cal D}^*Q_\Omega,
	\ee
$S_\Omega$ has been called dissipative current in \cite{Bel, Car}.

The partitioning of $\dot Q_\Omega$ into $I_{\partial \Omega}$ and $S_\Omega$ is, of course, gauge dependent. Models often offer a natural choice of $H$.  
\begin{figure}[htbp]
\begin{center}
\includegraphics[width=0.5\textwidth]{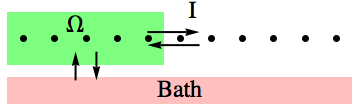}
\caption{The dots represent the system and the (green) box the subsytem $\Omega$. The bath is represented by the (red) box on the bottom. The horizontal arrow represent the current $I$  across the boundary of $\Omega$. The vertical arrows represent $S_\Omega$. }
\label{fig:box}
\end{center}
\end{figure}
\begin{exa}
Consider Fermions hopping on a one dimensional lattice which can also tunnel in and out of a bath.  The Lindbladian has
	\[
 H=\sum_j  \left(a^*_{j+1} a_j+a^*_ja_{j+1}-\mu a^*_{j} a_j\right), \quad \Gamma_j^-= \sqrt{\gamma_-} a_j, \quad \Gamma_j^+= \sqrt{\gamma_+} a^*_j
	\]
with $a_j$ the usual Fermion annihilation operators for site $j$. The charge in the left semi-infinite box is
	\[
	Q_L=\sum_{j\le 0} a^*_j a_j
	\]
and the currents in Eq.~(\ref{TotQ}) are 
\[
	I_{\partial L}=i(a^*_1a_0-a^*_0a_1), \quad 
S_L=2\sum_{j\le 0}(\gamma_+a_j a^*_j-\gamma_-a^*_ja_j):
\]
$I_{\partial L}$ is localized at the boundary of the box, whereas $S_L$ is not.
\end{exa}

In appendix \ref{sec:BCurrents} we elaborate on the notion of currents in magentic fields and in particular, describe the current densities in a model for the open quantum Hall system.  We find three current densities:  A Hamiltonian current density, a diffusion current and  a dissipative chiral current.



\section{Adiabatic Response }
\label{sec:adiabatic}


We are interested in adiabatically changing controls; $\phi=\phi(s)$ where  $s=\e t$ is the slow time. The  evolution equation for the state $\rho$ is
	\be\label{aeom}
	\e  \frac{d\rho}{d s} =\lin_\phi \rho.
	\ee
with initial state that is an instantaneous stationary state\footnote{The general case, where the initial state is not a stationary state, is different and more complicated, because  Eq.~(\ref{geometricgap}) has no analog. It is treated in e.g. \cite{Lidar}. }. 
The adiabatic time scale $\e^{-1}$ is 
the largest time scale in the problem and the adiabatic limit is governed by the stationary states of $\lin$. 
 Another natural,  limit, studied in \cite{DaviesSpohn, Pekolo}, lets the dissipation time scale increase with the adiabatic time scale. This limit is governed by the stationary states of $H$ rather than those of $\lin$.  

A key feature of adiabatic theory is that the evolution of $\rho$ is slaved to the evolution of $\sigma$. We borrow  from \cite{AFGG}:

\begin{prop}[Adiabatic evolution]
The solution of Eq.~(\ref{aeom}) with initial condition the stationary state $\sigma(0)$ is
 	\be\label{AdiabaticXpansion}
	(\pker \rho)(s)= \sigma(s)+\left\{\begin{array}{ll}  0&  \mbox{if $\dim\pker =1 $}\\ 
								O(\e)& \mbox{if $\dim\pker\ge 2$};
	\end{array}\right . 
	\ee
and 
	\be\label{AdiabaticXpansionQ}
 (\qran\rho)(s)=\e \lin^{-1}\dot\sigma(s) +O(\e^2),
	\ee
where $\sigma(s)$ is the corresponding integral of parallel transport.
\end{prop}
 $\lin^{-1}(\dot\sigma)$ is well defined and bounded since  $\dot\sigma\in \ran\lin$. This follows from parallel transport  $\dot\sigma=\qran\dot\sigma$  and the definition of   $\qran$ as the projection on $\ran\lin$.


\subsection{The response of unique stationary states}

We are interested in the response of the observable $X$ of an adiabatically driven system.  The case of   a unique stationary state is  simpler than the general case and we treat it first.

\begin{prop}[Response coefficients]\label{RC}
Suppose  
that the stationary state is unique.   Let $X$ be a bounded observable and $\rho$ a solution of the adiabatic Lindblad evolution, Eq.~(\ref{aeom}), with initial state a normalized stationary state $\sigma(0)$.  Then,  the response  at slow time $s$ is memory-less and is given by  
	\[
	 \tr\bigl(X\rho(s)\bigr)= \ \tr (X\sigma(\phi)) +\e  F_{\nu}(\phi)\,\dot\phi^\nu+O(\e^2) , \quad 
	\]
(summation implied)  with $\sigma(\phi)$ the instantaneous stationary state and $\phi=\phi(s)\in\cal M$. The {\em response coefficients} 
\begin{equation}
\label{eq:816}
 F_{ \nu}(\phi) = \tr \bigl(X (\lin^{-1}_\phi\pd_\nu \sigma(\phi))\bigr)
\end{equation}
are {\em functions} on control space ${\cal M} $.
\end{prop}

\begin{proof} 
This is a direct consequence of the adiabatic expansion, Eq.~(\ref{AdiabaticXpansion}).  
$\partial_\nu\sigma$ is trace class by Prop.~\ref{prop:trace}. \end{proof}

The first term $\tr(X\sigma)$ is  of $O(1)$, and describes the {\em persistent response}, a property of the  stationary state. The second term is the {\em driven response} which is proportional to the driving $\e \dot\phi$, the (unscaled) velocity of the controls.

 One is often interested in situations where $F(\phi)$  is constant on $\cal M$. This feature depends on additional structure (e.g. thermodynamic limit, disorder \cite{Bellissard, ASS90,AizenmanGraf}). Observe that the expression for $F(\phi)$  involves inverting $\lin$, an operator with a non-trivial kernel, and so is not completely elementary. 

\begin{rem}
The formula for the response coefficient $F_\nu$ can be cast in a way that is formally reminiscent of Kubo's formula:
	\[
	F_\nu(\phi)=-\lim_{\e \downarrow 0}\int_0^\infty dt\, \tr \bigl(( e^{(\lin^*_\phi-\e)t}  X)\ ( \partial_\nu\sigma(\phi))\bigr).
	\]
\end{rem}

When the manifold of stationary states is multidimensional, the persistent response has memory and $F$ can not be viewed  anymore as functions on control space $\cal M$  (see Section~\ref{sec:holonomy}). 
\section{Response of fluxes}\label{sec:fluxes}

 In the case of observable which are fluxes several simplifications occur:  There is no persistent response and the formula for $F_\nu$ simplifies and becomes elementary. If, in addition, the extremal stationary states are isolated  (Section~\ref{sec:holonomy}) then, in addition, $F$ defines a function on $\cal M$.

By definition, a flux (which is not explicitly time dependent,  Assumption~\ref{ass:partial}) can be written as
	\be
	\e \dot X= \lin^*X=\lin^*\qran^* X, \quad \Bigl( \dot X= \frac{dX}{ds}\Bigr),
	\ee
where the replacement $X$ by $\qran^*X$ relies on Eq.~(\ref{QL}) and is only of interest in  the infinite dimensional case where $X$ is unbounded while $\qran^*X$ is bounded.

\begin{thm}[Response of fluxes]\label{thm:fluxes}
Suppose 
  that $\sigma(\phi)$ is an integral of parallel transport and that the flux $\dot X$ and $\qran^*X$ are bounded operators.
Then, to leading order, the response is memory-less, linear in the driving and given by
	\[
	\tr\bigl(\dot X  \rho\bigr)(s)= F_{\nu}(\phi)\, \dot \phi^\nu + O(\e), \quad \Bigl(\dot X= \frac{dX}{dt}\Bigr)
	\]
(summation implied) and where the response coefficient
	\[
	F_{\nu}(\phi)=
   \tr \bigl((\qran^*X) \partial_\nu\sigma\bigr)
	\]
is a function on ${\cal M}$.
\end{thm}

\begin{proof}
From Eqs.~(\ref{comm}, \ref{AdiabaticXpansionQ}) we have
\[
\tr\bigl(\dot X  \rho\bigr)=\e^{-1}\tr\bigl((\lin^*\qran^*X)\rho\bigr)= \tr \left(\bigl(\lin^*\qran^*X\bigr) \lin^{-1}\dot\sigma\right)=
	 \tr \left((\qran^*X) \partial_\nu\sigma\right)\dot\phi^\nu. 
	\]
\end{proof}

The observation that one can sometimes avoid computing Green functions in linear response  is at the heart of the TKNN formula for the Hall conductance \cite{TKNN}. 
 

\subsection{Geometric magnetism for iso-spectral families}

The response coefficients of  an iso-spectral family generated by $G_\mu$ are naturally organized as a matrix $F_{\mu\nu}$, relating the response of the flux of $G_\mu$ to the driving $\dot\phi^\nu$. The  analog of the formula in Prop.~\ref{RC} is 
	\[
	\tr\bigl(\dot G_\mu  \rho\bigr)(s)= F_{\mu\nu}(\phi)\, \dot \phi^\nu + O(\e).
	\]
Combining Theorem~\ref{thm:fluxes} and Eq.~(\ref{ULocalResponse})  we get for the response matrix
\begin{align}\label{lie}
F_{\mu\nu}=\tr( (\qran^*G_\mu) \partial_\nu\sigma)
=-i\, \tr((\qran^*G_\mu) [G_{\nu},\sigma]).
\end{align}
We are now ready to state our {\em main result:}

\begin{thm}[Geometric response]\label{thm:LieResponse}
Suppose $G_\mu$ are bounded and $\sigma(\phi) = U(\phi) \sigma U^*(\phi)$ is an integral of parallel transport. Then the response matrix is antisymmetric and given by
\begin{equation} \label{lie2}
F_{\mu\nu}
= -i\tr([G_\mu,G_\nu]\sigma).
\end{equation}
If, moreover, $\sigma$ is a projection $P$
then $F$ is the adiabatic curvature of the bundle (see Appendix~\ref{appendix:a}):
	\be\label{chern}
	F_{\mu\nu}=i \,\tr\left({ P}_\perp[ \partial_\mu P,\partial_\nu P]\right).
	\ee

\end{thm}

For unitary evolutions, the first part of the theorem reduces  to a (special case of) result of Berry and Robbins \cite{BerryRobbins}, who coined the term  {\em geometric magnetism} for the anti-symmetric part of $F$. 

The second part of Theorem~\ref{thm:LieResponse}  extends the geometric interpretation of response matrix from the unitary case \cite{ASY87} to open systems. The conditions in the theorem are satisfied for Lindbladians representing relaxation to the ground state and dephasing Lindbladians whose initial state is a spectral projection.

\begin{proof}
The conditions have been set so that the formal manipulations are justified
\begin{align*}
F_{\mu\nu}=\tr( G_\mu \partial_\nu\sigma)
=-i\, \tr(G_\mu [G_\nu,\sigma])= -i\tr([G_\mu,G_\nu]\sigma),
\end{align*}
where in the second equality we used Eq.~(\ref{ULocalResponse}).

For the second part observe that the equation $-i[G,P] = \dot P$ implies
$$ P_\perp G P = iP_\perp \dot P \quad \mbox{and} \quad PGP_\perp = -i \dot P P_\perp .$$
Hence
\begin{align*}
F_{\mu \nu} &= -i\tr([G_\mu,G_\nu] P)  =-i \tr(P G_\mu P_\perp G_\nu P - P G_\nu P_\perp G_\mu P) \\
            &= -i\tr(\partial_\mu P P_\perp \partial_\nu P - \partial_\nu P P_\perp \partial_\mu P) \\
            &= i\tr(P_\perp [\partial_\mu P, \partial_\nu P]),
\end{align*}
where the first line is a readily checked identity.
\end{proof}


\begin{rem}
More details about the dephasing case are in Appendix~\ref{appendix:a} where a formula for response when $\sigma(\phi)$ is not the integral of parallel transport is given.
\end{rem}

In may happen that $[G_\mu,G_\nu]$ is proportional to the identity. 
The transport coefficients  are then purely geometric and independent of the dynamics. 
An example  is the Hall mobility:
\begin{exa}[Hall mobility]
\label{exa:landau2}
 The {\em Landau Lindbladian},  (see Eqs.~(\ref{landau}) and (\ref{bath}) of  Appendix \ref{sec:BCurrents})
is a  model of a quantum particle  in two dimensions, moving under the influence of a uniform, perpendicular, magnetic field, $B$,  in contact with a thermal bath. 
The corresponding controlled Lindbladian  is the iso-spectral family generated by 
\[
 G_{\mu} = B^{-1}\e_{\mu \nu} v_\nu,\quad [v_\mu,v_\nu]=i\e_{\mu\nu} B.
\]
The virtual work associated with the variation of the controls $\delta \phi^\mu$ is, as  explained in Appendix~\ref{sec:BCurrents}, the velocity $\dot x_\mu$, while  $\dot \phi^\nu$ is interpreted as an electric field. The transport coefficient relating  velocity to field strength is the mobility, given by
	\be
	F_{\mu\nu}=-i[G_\mu,G_\nu]= \frac{\e_{\mu\nu}}B
	\ee
and independent of $\gamma_\pm$ and the stationary state\footnote{ The Landau Lindbladian in the plane has $\dim \pker =\infty$. One can avoid this by considering the model   on the torus (with appropriate boundary conditions).}.

\end{exa}
By extension the model describes a gas of independent particles of density $\rho$ and conductance $\sigma_{\mu\nu}=\rho F_{\mu\nu}$. If the density corresponds to filling factor 1, i.e. to one particle per unit magnetic flux, then $\rho=B/2\pi$, and the conductance $\sigma_{\mu\nu}=(2\pi)^{-1}\e_{\mu\nu}$ is quantized in the same units as in the unitary case. 
\begin{rem}[Hall viscosity]
 Theorem~\ref{thm:LieResponse} can be used to recover, and generalize, results of Read and Rezayi \cite{Read} on the Hall viscosity:  viscosity is the study of the Landau Lindbladian under the iso-spectral family generated by shears.  The commutator of shears in two dimensions is the generators of rotation.
Theorem~\ref{thm:LieResponse} then relates the Hall viscosity with the expectation of the angular momentum per particle.
\end{rem}

\subsection{Friction and dissipation}
\label{sec:dissipation}

The fact that $F_{\mu\nu}$ of Eq.~(\ref{lie2}) is anti-symmetric does not imply the absence of dissipation. It only says that  looking at the response of fluxes is not appropriate for the study of dissipation.
To explain this statement consider the dissipation associated with the dragging of the anchoring point of a (damped)  oscillator coupled to a heat bath at velocity $\dot\phi$. The response coefficient relating force to velocity is friction. As there are three forces in the problem---the momentum rate, the force on the anchoring point, and the friction force---there are also three friction coefficients. The friction coefficient associated with the momentum rate vanishes, but the others do not.

\begin{exa} [Friction: Example~\ref{exa:oscillator} continued]
The  (unbounded) generator of shifts is, $p=i(a^*-a)/\sqrt 2$.  With $\sigma$ a thermal state of the oscillator, 
	$\partial_\phi\sigma= -i[p,\sigma]$
is trace class and Eq.~(\ref{lie2}) applies with $G_\mu=G_\nu=p$. 
The friction coefficient vanishes, as it must by anti-symmetry.
\end{exa}
The  momentum rate vanishes because it can not disentangle the heat lost to the bath from the mechanical work done by  the anchoring point.  To study dissipation it is not enough to look at the response coefficients of fluxes, nor is it enough to examine the energy of the system.

Indeed, the energy of the (small) system, in the adiabatic limit, is
	\be
	E= \tr (H\rho)= \tr (H\sigma) +O(\e )
	\ee
by Eq.~(\ref{AdiabaticXpansion}). In particular, for an iso-spectral family the energy is constant (to leading order) when $H$ and $\sigma$ undergo the same unitary transformation, as is the case  in the example of the damped oscillator. The energy  does not reveal the dissipation.

To reveal the dissipation one needs to look at the breakup of the energy to work and heat. The variation of the energy 
 \begin{equation}
\label{1.law}
   \delta\, \tr(H\rho)= \tr(H\,\delta \rho)+ \tr(\delta H\, \rho)
    \end{equation}
expresses the first law of thermodynamics \cite{LebowitzSpohn}
	\[
	\delta E= \delta W+\delta Q=\bigl(
\tr (\sigma\,\partial_\mu H )+\tr (H\,\partial_\mu \sigma )\bigr)\delta\phi+O(\e).
	\]
To compute the friction one needs to study the expectation of the spring force $-x$ rather than the momentum flux $\lin^*p$. (More generally, $\partial_\mu H$ rather than the flux  $\lin^*G_\mu$ of Eq.~(\ref{Q}).)

In general,  the computation of $\tr( \rho\, \partial_\mu H)$ is  complicated   for two reasons: First, one needs to evaluate $\lin^{-1}$. Second, in the case that the ground state is non-unique, it also needs the explicit expression for the $O(\e)$ term in the adiabatic expansion, Eq.~(\ref{AdiabaticXpansion}), which are history dependent. 
For dephasing Lindbladian such a computation is given in \cite{AFGK}. We shall not pursue this direction here.
 
As a sanity check, let us derive the first law of thermodynamics using the  tools of the previous sections.
Since $H$ is explicitly time-dependent Assumption~\ref{ass:partial} does not hold for $H$, its flux is now made of two terms:
	\be
	\dot H= (\partial_\mu H) \dot \phi^\mu + \lin^*H.
	\ee
Substituting in Theorems~\ref{RC}, \ref{thm:fluxes} indeed reproduces the first law: 
	\be
	\frac{d E}{dt}=\bigl( \tr( \sigma\, \partial_\mu H)+ \tr (H\, \partial_\mu\sigma)\bigr) \Bigl( \frac{d\phi^\mu}{dt}\Bigr) + O(\e^2).
	\ee


\section{Concluding remarks}

We have derived a simple and general formulas for the adiabatic response coefficients for  observable of the form $\dot X=\lin^*X$. In the case of iso-spectral families of Lindbladians, the response matrix is determined by geometry and is purely anti-symmetric. We find a range of circumstances where the response coefficients are given by the adiabatic curvature of the associated stationary projections. It will be interesting to  extend the theory  to  models of extended systems with (non-interacting) fermions.

{\bf Acknowledgments.}
JEA is supported by the ISF, the NSF under Grant No. PHY11-25915 and the fund for promotion of research at the
Technion. 
 MF was supported by UNESCO and ISF. We thank M. Porta for useful discussions.

\appendix

\section{Currents in a magnetic field}\label{sec:BCurrents}
%

The action of a magnetic field on charged particles endows the dynamics with chirality. This has interesting consequences for currents. 
Consider the Lindbladian describing a charged particle in the plane under the influence of a constant magnetic field coupled to a heat bath. 
 The Hamiltonian is the Landau Hamiltonian 
	\be\label{landau}
	H= D^* D,  \quad D= -i\partial_1+\partial_2 + B x_2\equiv v_1+iv_2, \quad(v_\mu=v^*_\mu) ,
	\ee	
and the thermal bath is represented by  (cf.  Example~\ref{exa:decay})
\be\label{bath}
\Gamma_-=\sqrt{\gamma_-} D, \quad \Gamma_+ = \sqrt{\gamma_+} D^*, \quad (\gamma_->\gamma_+\ge 0).
\ee

We shall call the generator of the corresponding evolution a thermal Landau Lindbladian. The model has a current density associated to the total current Eq.~(\ref{TotQ}).
\begin{prop}
The  (total) current density of the Landau Lindbladian of Eqs.~(\ref{landau}, \ref{bath}) is
\begin{equation}
\label{cudensity}
 j_\mu(x_0) = \{\rho(x_0), v_\mu\}- (\gamma_+ + \gamma_-) \partial_\mu \rho(x_0)  + (\gamma_- - \gamma_+)  \e_{\mu \nu} \{\rho(x_0),v_\nu\} .
\end{equation}
The charge density is $\rho(x_0)=\delta(\cdot-x_0)$ and $\e_{\mu\nu}$ is the completely anti-symmetric (Levi-Civita) tensor. 
 The (total) current satisfies charge conservation:
\be\label{conservation}
\partial_t \rho=-\partial_\mu j_\mu.
\ee

\end{prop}
Before proving the statement, let us comment about its content. By Eq.~(\ref{conservation}) the total current is unique up to a curl.  The Hamiltonian current is proportional and parallel to the velocity $2v_\mu$. The dissipative current has a (non-chiral) diffusive term proportional to the gradient of the density and a further chiral term. The dissipative currents can be interpreted in terms of Brownian motion (see below).

\begin{proof}
The dissipative terms of the Lindbladian are
$$
\mathcal{D}^*_\pm(X) = \Gamma_\pm^*[ X ,\Gamma_\pm]+ [ \Gamma_\pm^*, X ]\Gamma_\pm.
$$
For a function $X = f(x_1,x_2)$ of position we have
\begin{equation}
\label{xspace}
i[H,f]=\{v_\mu,\partial_\mu f\},
\quad
\mathcal{D}^*_\pm(f) = \gamma_\pm  \laplace f \mp 2\gamma_\pm (\partial_\mu f) \e_{\mu \nu} v_\nu.
\end{equation}
Eq.~(\ref{Irate}) then gives the dual form of the statements of the proposition, namely,
\[
\partial_t f =\int j_\mu(x_0)\partial_\mu f(x_0)d^2x_0,\quad f=\int \rho(x_0)f(x_0)d^2x_0.
\] 
\end{proof}
\subsection{Stochastic interpretation}
The dissipative currents  admit an interpretation in terms of a (classical) stochastic process.
To see this note first that
 for functions $X = f(v_\mu)$ of either velocity ($\mu=1,2$) 
\begin{equation}
\label{pspace}
\mathcal{D}^*_\pm(f) = \gamma_\pm B^2 f'' \pm 2\gamma_\pm B f' v_\mu 
\end{equation}
which can be read as if originating from a (exciting or damping) Langevin equation
$$
dv_\mu = \pm 2\gamma_\pm  B v_\mu dt + B db_{\mu,t},
$$
where $b_\mu$ is a Brownian motion with zero drift and variance
$$
\mathbb{E}(db_{\mu,t} db_{\nu,t}) =2\gamma_\pm  \delta_{\mu \nu} dt.
$$
In fact, expanding $\mathbb{E}(f(v_\mu+dv_\mu))$ to first order in $dt$ and to second order in $db_t$ yields that expression.

To derive the Langevin equation for $dx$ we first note that the guiding center $(r_1,r_2)$,
\begin{equation}
\label{guiding}
r_\mu = x_\mu + \frac{\e_{\mu \nu}}{B} v_\nu
\end{equation}
satisfies $[r_\mu,\,v_\nu] = 0$ and thus is a constant of motion for the Lindbladian, $\lin^*r_\mu =0$. 
Insisting on $r_\mu$ being a constant of motion, we have
\begin{equation}
\label{xlangevin}
d x_\mu = -\frac{\e_{\mu \nu}}{B} d v_\nu = 
\e_{\mu \nu}( \mp 2\gamma_\pm v_\nu dt - db_{\nu,t}).
\end{equation}
 In view of $\mathbb{E}(\e_{\mu \nu} db_{\nu,t} \e_{\mu' \nu'} db_{\nu',t}) =2\gamma_\pm \delta_{\mu \mu'}dt$ this is the Langevin equation corresponding to  Eq.~(\ref{xspace}). (Beware: $d x_\mu\neq v_\mu dt$.)
%
We can now combine $\lin^*r_\mu =0$ with Theorem~\ref{thm:VW} to conclude
\begin{prop}[Velocity as virtual work]\label{IsoLandau}
The (negative) virtual work associated with the variation $\delta\phi^\mu$ of the iso-spectral family of Landau Lindbladians generated by 
	\be\label{GLH}
	 G_{\mu} = B^{-1}\e_{\mu \nu} v_\nu.
	\ee
is the velocity $\dot x_\mu$. 			
\end{prop}

$G_\mu$ is the generator of the unitary family $U(\phi)$ given by
	\be
	(U_\phi\psi)(x_1,x_2)= \psi(x_1+\phi_2/B,x_2-\phi_1/B) e^{i\phi_2 x_2}.
	\ee
The physical interpretation of the controls emerges by noting that 
	\be\label{fluxes}
         U_\phi v_\mu U^*_\phi = v_\mu - \phi_\mu.
	\ee
Since $\phi_\mu$ appears in $H$ like a pure gauge field, its variation in time, $-\dot\phi$ is a constant electric field that drives the system. 
\begin{rem}[Gauge covariance]
The proposition may be viewed as a manifestation of
gauge and translation covariance, in the sense that $-i\partial_\mu$
and $x_\mu$ appear in the Lindbladian only through the minimal coupling expression $v_\mu$. The virtual work associated with the variations  generated by $-G_\mu$ is the same as the variation generated by $x_\mu$. This follows from 
\[
-\frac{\partial}{\partial\phi_\mu}U_\phi v_\nu U^*_\phi=i[v_\nu, x_\mu];
\]
in fact both sides equal $\delta_{\mu\nu}$.
\end{rem}

\section{Geometry of projections}\label{appendix:a}

Consider  continuous orthogonal projections $P_j(\phi)$ with $\sum P_j(\phi) = \id$. 
The superprojection that takes $\rho$ to $\ran \pker = \mathrm{Span} \{ P_j \}$ is, Eq.~(\ref{PD}),
$$
\pker\rho= \sum_j P_j \rho P_j.
$$
We are going to describe parallel transport inside $\ran \pker$ \cite{Kato50}.

For a given path $P_j(t)$, parallel transport $\dot u=\dot P_jP_j u$ maps 
vectors $u(0)$ in the range of $P_j(0)$ to vectors $u(t)$ in that of $P_j(t)$. That map $U(t)$ is unitary and generated by 	
\be\label{K}
	K:=i \dot U U^*= \sum_j A_j, \quad A_j= i\dot P_{j}P_j.
	\ee
In fact $K^*=K$, since $A_j^*=-iP_j\dot P_{j}=-i\dot P_{j}(1-P_j)=-i\dot P_{j}+A_j$, and, for $U$ so defined, $u(t)=U(t)u(0)$ satisfies 
\[
\dot u=-iKu=\sum_j\dot P_{j}P_ju,
\]
as required. And for $\sigma(t) = U(t) \sigma U^*(t)$ the parallel transport equation~(\ref{parTran}), $\pker(t) \dot \sigma(t) = 0$, holds true.

When $\mathrm{dim} P_j = 1$, the parallel transport is manifestly path independent.
In general, this is determined by the standard condition of vanishing curvature:
\begin{prop}\label{integrability}
Let 	${\cal A}= (d \pker) \pker$ be an operator valued 1-form.  The differential equation	
\[
d\sigma ={\cal A} \sigma
\]
admits a (locally path independent) solution $\sigma$ if and only if the curvature vanishes
	\be\label{Curv}
	{\cal R}=0,\quad {\cal R}=-i\pker \,d\pker \wedge d\pker \, \pker. 
	\ee
\end{prop}
 It implies the following criterion for the case of parallel transport of projections.

\begin{prop}[Adiabatic curvature]
\label{PathProp}
The parallel transport constructed above is locally path independent if and only if the {\em adiabatic curvature} 
 	\[
	R_{\mu \nu} := \partial_\mu K_\nu -\partial_\nu K_\mu +i [K_\mu,K_\nu]=-i \sum_{ j} {P _j}[ \partial_\mu P_j,\partial_\nu P_j]
	\]
commutes with all elements in $\ran \pker$.
\end{prop}

\begin{proof} 
Parallel transport of a vector $u_j$ in the range of $P_j$ along an infinitesimal square $d\phi^\mu d\phi^\nu$ maps 
 $$ u_j \quad \to \quad u_j - i R_{\mu \nu} u_j \d\phi^\mu d\phi^\nu +o(d\phi^2).$$ The associated adjoint transformation maps the state $\sigma = P_j\sigma P_j $ as
 $$
 \sigma \quad \to \quad \sigma -i[R_{\mu \nu}, \sigma] \d\phi^\mu d\phi^\nu +o(d\phi^2).
 $$
This allows to read off the curvature ${\cal R}_{\mu \nu}$ of $\pker$ seen in Eq.~(\ref{Curv}): By ${\cal R}={\cal R}\pker$ we have 
\[
{\cal R}_{\mu \nu}\rho=[R_{\mu \nu}, \pker\rho].
\]
Hence the criterion of vanishing curvature states that $R$ commutes with all elements in $\ran \pker$.

Computation gives the commutator
	\begin{align}
	  [K_\mu,K_\nu]&
	    =\sum_{ j} {P _j}_\perp[ \partial_\mu P_j,\partial_\nu P_j]
	  \end{align}
as the sum of adiabatic curvatures of all the spectral projections. 
Since
	\[
	\partial_\mu K_\nu-\partial_\nu K_\mu=-i \sum_j [\partial_\mu P_j,\partial_\nu P_j]
	\]
one finds
	\[ 
\partial_\mu K_\nu -\partial_\nu K_\mu +i [K_\mu,K_\nu]=-i \sum_{ j} {P _j}[ \partial_\mu P_j,\partial_\nu P_j].
	\]
\end{proof}

For an iso-spectral family of projections $$P_j(\phi) = \exp(-i G \phi) P_j(0) \exp(i G \phi)$$ the generator of parallel transport, Eq.~(\ref{K}), is
$$K =G - \sum_j P_j G P_j = G - \pker^*(G), $$
since $\pker^*=\pker$ by Eq.~(\ref{PD}). While it does not coincide with $G$ it differs from it only inside $\ran \pker$
 \be \label{lastone}
\qran^*(G)  = \qran^*(K).
\ee

When $\mathrm{rank}\, P_j > 1$ and $\mathrm{dim}\, {\cal M} > 1$ the parallel transport can not be integrated in general. And the response coefficients are not functions on the manifold.

\begin{thm}
Suppose $G_\mu$ are bounded and $\lin$ is a dephasing Lindbladian. Then the 
response associated to the driving path $\phi(s)$ and flux $\dot G_\mu$ 
depends only on the integral of parallel transport $\sigma(\phi)$ and the 
derivative $ \delta \phi$ at the end point, 
$$
 \tr(\lin^*(G_\mu) \rho(s)) = F_{\mu \nu} \pd_\nu \phi(s),
$$
where
$$
F_{\mu \nu} = i\tr([K_\mu,K_\nu] \sigma).
$$
\end{thm}

\begin{proof}
\begin{align*}
F_{\mu\nu}&=\tr( Q^*(G_\mu) \partial_\nu\sigma) \\
&=i\, \tr(\qran^*(K_\mu) [K_\nu,\sigma])= i\tr([K_\mu,K_\nu]\sigma),
\end{align*}
where the second line express parallel transport and uses Eq.~(\ref{lastone}).
The last equality is by $\pker[K_\nu,\sigma]=0$, which characterizes parallel transport and by the way restates Eq.~(\ref{consistencyG}).
\end{proof}
\begin{exa}[Taming $X$: Example~\ref{taming} continued]
Consider a family of Hamiltonians generated by a momentum shift
$$
e^{-i \phi X} H e^{i \phi X} = e^{i \phi} T \otimes a + e^{-i \phi} T^* \otimes a^*  = P_+(\phi) - P_-(\phi),
$$
where $X$ is the position operator. The generator of the parallel transport is 
\begin{align*}
K = i( \dot{P}_+P_+ + \dot{P}_-P_- )
  =-\frac{i}{4} [H,\dot{H}] = a^* a - \frac{1}{2}.
 \end{align*}
Although $K \neq X$, their difference commutes with the Hamiltonian.
Furthermore $K$ intertwines $P_\pm$,
$$
KP_+ = P_-K,
$$
which is equivalent to the statement that $K$ generates no motion inside $\ran P_\pm$, $P_\pm K P_\pm = 0$.
\end{exa}

\section{Currents and unbounded observables}\label{unbounded}

We discuss the precise meaning of Eq.~(\ref{X0atEqUb}) when $X$ is unbounded. We still assume that $H$ and $\Gamma_\alpha$ are bounded, while $X=X^*$ need not be. Yet, the commutators $[H,X]$ and $[\Gamma_\alpha, X]$, defined as quadratic forms on the domain $D(X)$ of $X$, are assumed bounded, and $\Gamma_\alpha D(X)\subset D(X)$. Then $\dot X=\lin^*(X)$ is a bounded operator by natural interpretation of Eq.~(\ref{Heisenberg}) in the sense of quadratic forms.

\begin{prop} \label{NoCurrentsUnb}
Under the stated conditions, $\tr(\dot X\sigma) =0$ for any (trace class) stationary state $\sigma$. Moreover, $Q^*(X)$ is well-defined as a bounded operator. It is given as a strong limit, $Q^*(X)=\slim_{n\to\infty} Q^*(X_n)$, by means of any sequence of bounded approximants $X_n$ with $X_n\varphi\to X\varphi$, ($\varphi\in D(X)$); finally $\lin^*(X)=\lin^*(Q^*(X))$.\end{prop}

\begin{proof} There exist sequences $X_n$ as stated, e.g. $X_n=X/(1+n^{-1}X^2)$. The assumption states that the bounded operator $[H,X]$ is characterized by the property 
	\be
 (\varphi, [H,X]\psi)=(H\varphi, X\psi)-(X\varphi, H\psi),\qquad
(\varphi,\psi\in D(X))\,.
\ee
Hence $[H,X_n]\to [H,X]$ (weakly), and similarly for $[\Gamma_\alpha, X]$. Thus $\lin^*(X_n)\to\lin^*(X)$ (weakly). Using that $A_n\to A$ (weakly) and $B$ trace class imply $\tr(A_nB)\to \tr(AB)$, we conclude 
\be
\tr(\lin^*(X)\sigma)=\lim_{n\to\infty}\tr(\lin^*(X_n)\sigma)=0
\ee
by Eq.~(\ref{X0atEq}). Moreover, by Eq.~(\ref{QL}) we have $\lin^*\qran^*(X_n)\to \lin^*(X)$. We notice that $\lin^*$ is weakly continuous, and so are $(\lin^*-z)^{-1}$ and the inverse of $\lin^*$ on $\ran\lin^*=\ran\qran^*$, i.e. 
\be
(\lin^*)^{-1}=-\frac{1}{2\pi i} \oint\frac {dz}  {z}(\lin^* -z)^{-1}
\ee
in the notation of Eq.~(\ref{PKer}). As a result $\qran^*(X_n)$ is weakly convergent to a limit denoted $\qran^*(X)$, and the result follows. \end{proof}


\bibliography{lind}
\bibliographystyle{plain}

\end{document}